\documentclass{ifacconf}

\makeatletter
\let\old@ssect\@ssect
\makeatother

\usepackage{natbib}
\usepackage{hyperref}
\makeatletter
\def\@ssect#1#2#3#4#5#6{%
  \NR@gettitle{#6}
  \old@ssect{#1}{#2}{#3}{#4}{#5}{#6}
}
\makeatother

\usepackage{graphicx}      
\usepackage{natbib}        
\usepackage{amsmath, amssymb}
\usepackage{bm}
\usepackage{xcolor}
\usepackage{stfloats}  
\usepackage{algorithm}
\usepackage{algpseudocode}

\fnbelowfloat          
\usepackage{hyperref}

\newtheorem{assumption}{Assumption}
\newtheorem{problem}{Problem}
\newtheorem{theorem}{Theorem}
\newtheorem{lemma}{Lemma}
\newenvironment{proof}{
  \par\noindent\textbf{Proof.}\ }
  {\hfill$\square$\par}
\begin{document}
\begin{frontmatter}


\title{Supervisory Measurement-Guided Noise Covariance Estimation: Discussing Forward and Reverse Differentiation} 


\author[First]{Haoying Li} 
\author[First]{Yifan Peng} 
\author[Third]{Yuchi Wu}
\author[First,Second]{Junfeng Wu}
\address[First]{School of Data Science, the Chinese University of Hong Kong, Shenzhen 
   (e-mail: haoyingli,yifanpeng@link.cuhk.edu.cn, junfengwu@cuhk.edu.cn).}
\address[Second]{School of Artificial Intelligence, the Chinese University of Hong Kong, Shenzhen.  }
\address[Third]{School of Mechatronic Engineering and Automation,
Shanghai University(e-mail: wuyuchi@shu.edu.cn)}


\begin{abstract}             
Reliable state estimation depends on accurately modeled noise covariances, which are difficult to determine in practice.
This paper formulates the noise covariance estimation as a bilevel optimization problem that factorizes the joint likelihood of primary and supervisory measurements to reconcile information exploitation with computational tractability. 
The factorization converts the nested Bayesian dependency into a Markov-chain structure, allowing efficient computation. 
At the lower level, a Kalman filter with state augmentation performs such computation.
Meanwhile, closed-form forward and reverse differentiation provide efficient gradients for the upper-level updates, and we compare the two modes’ space and time complexities to inform their practical selection.
The upper level subsequently refines the noise covariances to guide the lower-level estimation. 
Taken together, the proposed algorithms offer a systematic and computationally efficient approach to noise covariance estimation in linear Gaussian systems.
\end{abstract}

\begin{keyword}
Noise covariance estimation; Bayesian inference; Kalman filtering.
\end{keyword}

\end{frontmatter}

\section{Introduction}
State estimation is central to a broad spectrum of engineering applications, such as localization, tracking, and feedback control. 
The accuracy of a state estimator relies critically on precise modeling of the system and its associated noise statistics.
However, in practice, the measurement noise covariances are often unknown, time-varying, or difficult to characterize reliably~\cite{huang2017novel}, 
leading to trial-end-error tuning.
This motivates automatic estimation of noise covariances, with the parameters viewed as optimization variables inferred from noisy data.
The noise covariance estimation problem has been widely studied, and we next provide a brief overview of existing approaches organized by their underlying optimization objectives.

Many works optimize covariances by directly assessing state-estimation performance.
The paper~\cite{wang2023neural} introduces a disturbance estimation for moving-horizon estimation, formulated as a bilevel problem to reduce the state–reference deviation.
The paper~\cite{qadri2024learning} formulates covariance learning as a bilevel problem that minimizes estimation error against ground-truth states.
While empirically effective, their dependence on ground-truth data and lack of probabilistic structure may hinder the recovery of coherent covariance estimates.

On the other hand, parameter inference can be formulated through maximum likelihood (ML) or maximum a posteriori (MAP).
Classical MLE-type approaches include gradient-based optimization and derivative-free schemes such as expectation maximization (EM).
For example, \cite{yoon2021unsupervised} proposes an EM-based covariance learning method, but its reliance on full posterior approximation limits its practicality.
Numerous gradient-based methods are available.
In~\cite{khosoussi2025joint}, joint MAP estimation of the states and noise covariances is solved using block coordinate descent, with most of the computational cost arising from the state-update step.
Alternatively, one may maximize the marginal likelihood of the parameters by integrating out the states, which can be computed efficiently via Kalman filtering.
In Kalman filtering, likelihood sensitivities with respect to the parameters are typically propagated with the state estimate~\cite{sarkka2023bayesian}, and ~\cite{tsyganova2017svd} focus on improving the numerical stability of these recursions.
More recently,~\cite{parellier2023speeding} derives closed-form parameter derivatives through the Kalman filter via backpropagation, enabling faster gradient computation.
 
Existing noise-covariance estimation methods are well supported theoretically and empirically, motivating our use of an MLE/MAP formulation, which provides a statistically grounded framework.
Two issues arise in this context.
One concerns the use of higher-fidelity information to improve parameter estimation, for which we introduce supervisory measurements and develop a probabilistic factorization that yields a bilevel optimization structure. 
The other concerns the efficiency of likelihood and derivative computation, for which we rely on a Kalman filter and analyze forward- and reverse-mode differentiation, characterizing their computational and memory profiles and providing guidance on their practical deployment.
The contributions of this work are summarized as follows.
\begin{enumerate}
\item \textbf{Likelihood factorization for parameter MLE.} We introduce a likelihood factorization that separates primary measurements, whose noise covariances are to be estimated, from higher-fidelity supervisory measurements used to assess parameter quality.
This yields a bilevel optimization structure: a Kalman filter processes the chain-structured primary measurements at the lower level, while the upper-level objective integrates both primary and supervisory losses to fully exploit the information available in the system.
\item \textbf{Comparison of forward- and reverse-mode analytical derivative computation.}
We derive closed-form gradients of the proposed MLE formulation under both forward and reverse differentiation, and compare their behavior within the Kalman-filter–based lower-level solver in terms of computational and memory complexity. Practical guidelines for selecting between the two modes are provided.
\end{enumerate}

\section{Problem Formulation}
Consider a time-varying linear system
\begin{align}
 x_k & = F_k x_{k-1} + B_k u_k 
 +w_k,\quad &&w_k\sim \mathcal{N}(0,Q_k(\theta))\label{eq:system}\\
  y_k & = H_k x_k + \nu_k, \quad &&\nu_k\sim \mathcal{N}(0,R_k(\theta))\label{eq:primary_mea}
\end{align}
where the state $x_k \in \mathbb{R}^d$, the control input $u_k \in \mathcal{U}$, and $Q_k(\theta)$ and $R_k(\theta)$ are 
unknown noise covariances parameterized by $\theta \in \Theta$, with $\Theta$ denoting an admissible parameter set (e.g., ensuring positive definiteness).
Let $\bm{y}^o_{1:k}\triangleq \{y_i\}_{i=1}^k$, collectively referred to as  
\emph{primary measurement}, whose uncertainty is unknown and to be estimated.

In addition to the primary measurement, the system also incorporates \emph{supervisory measurements} with low or known uncertainty, which provide reliable information for evaluating and improving the parameter estimates.
These supervisory measurements are temporally sparse and only available at selected time steps.

For clarity of presentation, we introduce the following notation.
Denote the set of states up to time $k$ as \(\bm{x}_{1:k} = \{x_1, \ldots, x_k\}\).
Let $\bm{x}^s_{1:k}$ be the subset of states involved in the supervisory measurements
and define the corresponding stacked state vector
$$X^s_k = [(x^s_{i_k,1})^\top\;\cdots\; (x^s_{i_k,l})^\top]^\top.$$
The remaining states are given by $\bm{x}^o_{1:k}=\bm{x}_{1:k} \setminus \bm{x}^s_{1:k}$. 
For a time horizon of length \(N\), the subscript \(1:N\) is omitted when clear from context.

The supervisory model, interpreted as observing a subset of states over the entire time horizon, is then given by
\begin{equation}\label{eq:super_mea}
    \bm y^s = H^s X_N^s + \nu^s, \nu^s\sim\mathcal{N}(\bm{0},\Psi),
\end{equation}
where $\Psi$ denotes the supervisory measurement noise covariance, which is assumed to be known.

Two typical supervisory measurements are provided in Fig.\ref{fig:SupervisoryExample}. 
The panel on the left illustrates a loop closure, where the trajectory returns to a previously visited location.
The right panel shows ground-truth states, commonly used in deep learning–based Kalman filtering for supervision as in~\cite{zhao2019learning, revach2022kalmannet}.
Such measurements are typically obtained from high-precision instruments and allow one to set $\Psi=\mathbf 0$ in~\eqref{eq:super_mea}, corresponding to the Dirac delta limit of Gaussians with vanishing variance.
\begin{figure}[ht]
    \centering
    \includegraphics[width=0.8\linewidth]{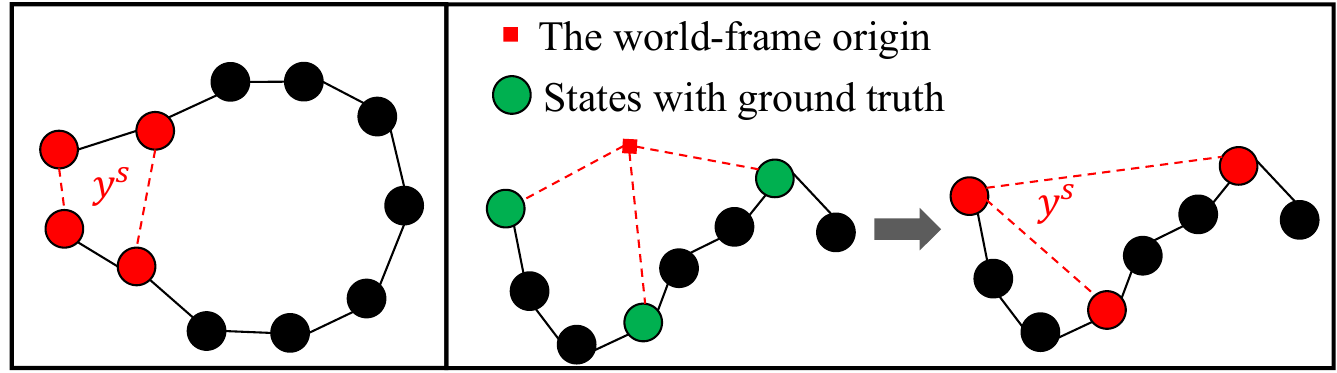}
    \caption{Examples of the supervisory measurements. The states represented by red circles are states in $\bm x^s$.}
    \label{fig:SupervisoryExample}
\end{figure}

In accordance with common practice in state estimation, we assume the following independence condition.
\begin{assumption}~\label{ass:indep}
The noises $w_k$'s, $\nu^b_k$'s, $\nu^v_k$'s, and $\nu^s$'s are mutually independent random variables.\qed
\end{assumption}

We now formalize the covariance estimation problem as an MLE problem as follows
\begin{equation}\label{eqn:MLE_problem}
\hat{{\theta}} = \operatorname*{argmax}_{{\theta}\in\Theta} p(\bm{y}^o,\bm y^s \mid \theta),
\end{equation}
where $\hat{\theta}$ is termed the MLE estimate to~$\theta$. 
Note that the MLE of the parameters can be readily extended to an MAP estimate by introducing an appropriate prior $p(\theta)$.

In what follows, we solve the MLE for the parameter formulated in~\eqref{eqn:MLE_problem} by the gradient-descent method.
Section~\ref{sec:factorization} presents a probabilistic factorization leading to a tractable objective.
Section~\ref{sec:state filter} introduces a Kalman-filter–based forward pass for efficient likelihood computation.
Sections~\ref{sec:forward_diff} and~\ref{sec:reverse_diff} develop forward- and reverse-mode differentiation for obtaining the gradients.
\section{ Likelihood and its Factorization}
\label{sec:factorization}

\subsection{Likelihood Factorization}
Our analysis begins with the factorization of the likelihood in~\eqref{eqn:MLE_problem}.
Treating the states as latent variables and applying Assumption~\ref{ass:indep}, the likelihood can be factorized as
\begin{align*}
&p(\bm{y}^o,\bm{y}^s \mid \theta) \\
=& \int \int p(\bm{y}^o,\bm{y}^s,\bm{x}^o,\bm{x}^s \mid \theta)\, \mathrm{d} \bm{x}^o \mathrm{d} \bm{x}^s \\
= &\int \int 
p(\bm{y}^s \mid \bm{x}^o, \bm{x}^s, \bm{y}^o, \theta)\, p(\bm{x}^o, \bm{x}^s, \bm{y}^o \mid \theta)
\, \mathrm{d} \bm{x}^o \mathrm{d} \bm{x}^s \\
=& \int \int 
p(\bm{y}^s \mid \bm{x}^s)\, p(\bm{x}^o, \bm{x}^s \mid \bm{y}^o, \theta)\, p(\bm{y}^o \mid \theta)
\, \mathrm{d} \bm{x}^o \mathrm{d} \bm{x}^s,
\end{align*}
where the integrations are taken over some product spaces of 
$\mathbb R^d$ by default. 
Further marginalizing out $\bm{x}^o$ yields
\begin{align*}
&p(\bm{y}^o,\bm{y}^s \mid \theta) \\
=& p(\bm{y}^o \mid \theta) \int  
p(\bm{y}^s \mid \bm{x}^s)\,\left(\int p(\bm{x}^o, \bm{x}^s \mid \bm{y}^o,\theta)\,
\mathrm{d} \bm{x}^o \right)\mathrm{d} \bm{x}^s \\
=& p(\bm{y}^o \mid \theta) \int 
p(\bm{y}^s \mid \bm{x}^s)\, p(\bm{x}^s \mid \bm{y}^o, \theta)\,
\mathrm{d} \bm{x}^s.
\end{align*}
By taking the negative logarithm of the likelihood, the objective function becomes
\begin{equation}\label{eq:NLL}
\mathcal{L}(\theta)\triangleq-\log p(\bm{y}^o,\bm{y}^s \mid \theta) 
= \ell^o(\theta) + \ell^s(\theta),
\end{equation}
where 
\begin{align}
    \ell^o(\theta) &\triangleq -\log p(\bm{y}^o \mid \theta),\\
    \ell^s(\theta)& \triangleq -\log \int p(\bm{y}^s \mid \bm{x}^s)\, p(\bm{x}^s \mid \bm{y}^o,\theta)\,\operatorname{d} \bm{x}^s\label{eq:ls}.
\end{align}
We refer to $\ell^o(\theta)$ as the \emph{primary loss} and $\ell^s(\theta)$ as the \emph{supervisory loss}.
By the probability chain rule,  the following decomposition holds
\begin{equation}\label{eq:lo1}
    \ell^o(\theta)
    = -\sum_{k=1}^N \log p\!\left(y_k \mid \bm{y}^o_{1:k-1}, \theta\right),
\end{equation}
where each term satisfies
\begin{equation}\label{eq:decom}
    p\!\left(y_k \mid \bm{y}^o_{1:k-1}, \theta\right)
    = \int p\!\left(y_k \mid x_k, \theta\right)
      p\!\left(x_k \mid \bm{y}^o_{1:k-1}, \theta\right)
      \mathrm{d}x_k.
\end{equation}
Here, $p(y_k \mid x_k, \theta)$ is the primary measurement model~\eqref{eq:primary_mea}.

Evaluating $\mathcal{L}(\theta)$~\eqref{eq:NLL} given $\theta$ requires the computation of $p(x_k \mid \bm{y}^o_{1:k-1},\theta)$ and 
$p(\bm{x}^s \mid \bm{y}^o,\theta)$. 
These quantities are derived in the next section.

\subsection{State Filter for Likelihood Computation}\label{sec:state filter}
As both conditional distributions rely on the sequential primary measurements $\bm y^o$, a Kalman filter provides an efficient mechanism for their computation. 
To explicitly maintain the correlations among the states in $\bm{x}^s_{1:k}$, we introduce the augmented state
$$
X_k=[x_{k}^\top\;
x_{i_{k,1}}^\top\;\ldots\;x_{i_{k,l}}^\top
]^\top=[x_{k}^\top\;(X^s_k)^\top]^\top.
$$

Let $\bar{X}_k,\bar{P}_k$ and $\hat{X}_k,\hat{P}_k$ denote the prior and the posterior estimates with the estimation error covariance of $X_k$. 
The prediction step is as follows:
\begin{equation}
\bar{X}_k=   {F}_{k,0} \hat{X}_{k-1} +  B_{k,0}  u_k, \;\bar{P}_k={F}_{k,0} \hat{P}_{k-1}  {F}_{k,0}^\top  + Q_{k,0},\label{eq:predx}
\end{equation}
where $ F_{k,0} = \bigl[\begin{smallmatrix}
    F_{k} & \bm 0 \\ \bm 0 & I
\end{smallmatrix}\bigr]$, $B_{k,0}=\bigl[\begin{smallmatrix}
    B_{k}  \\ \bm 0 
\end{smallmatrix}\bigr]$ and $ Q_{k,0} = \bigl[\begin{smallmatrix}
    Q_k & \bm 0 \\\bm 0 &\bm 0
\end{smallmatrix}\bigr]$, $\bar P_k =\bigl[\begin{smallmatrix}
   \bar P^o_k & \bar P_k^{os}\\
   \bar P_k^{so} &\bar P_k^s
\end{smallmatrix}\bigr]$.

The update step serves to incorporate primary measurements as follows
\begin{align}
r_k = &y_k - H_{k,0} \bar{X}_k = y_k-H_k \bar x_k ,\label{eq:r}\\
S_k =& {H}_{k,0} \bar{P}_k {H}_{k,0}^\top + \Sigma_k = H_k \bar{P}^o_k H_k^\top +R_k,\label{eq:S}\\
K_k =&\bar{P}_k H_{k,0}^\top S_k^{-1},~~\check{P}_k = \bar{P}_k-K_k H_{k,0} \bar{P}_k\label{eq:K}\\
\check{X}_k =& \bar{X}_k + K_kr_k,\label{eq:updateY}
\end{align}
where $H_{k,0}= [H_k ~\bm{0}]$. 

Once the update completes, $\hat{x}_k$ will append to the end of $\hat{X}_k$ if $x_k \in \bm{x}_k^s$, resulting in the covariance update
\begin{equation}\label{eq:copyP}
 \hat{X}_k = J_k \check{X}_k, \quad  \hat{P}_k=J_k \check{P}_k  J_k^\top,
\end{equation}
where $J_k = [I~ J^\top]^\top$
with {$J = [I \;  0]$ } when the appending takes place, and 
$J_k = I$, otherwise.

For notational simplicity, let $\hat{\bm X}(\theta)$ and $\hat{\bm P}(\theta)$ denote the collections of
${\hat{X}_N^s,,\hat{P}_N^s}$ and ${\bar{x}_k,,\bar{P}^o_k}$.
Given the above factorization, the noise covariance estimation problem can be formulated as a bilevel optimization problem:
\begin{problem}
\begin{align*}
&\text{\small (Upper level)} ~~
\hat{\theta} = \arg\min_{\theta \in \Theta}
\mathcal{L}(\hat{\bm{X}}(\theta), \hat{\bm{P}}(\theta)), \\[0.3em]
&\text{\small (Lower level)} ~~ \text{s.t.}
\{\hat{\bm{X}}(\theta), \hat{\bm{P}}(\theta)\} =
\operatorname*{argmax}\limits_{x_k,\, 1 \leq k \leq N}  
p(x_k \mid \bm{y}^o_{1:k}, \theta).
\end{align*}
\end{problem}

The following lemmas derive the loss terms in the upper-level problem.
Let $\cong$ denote equivalence up to some additive constants.
 
\begin{lemma}[Computation of $\ell^o(\theta)$]\label{lem:llo}
The $\ell^o(\theta)$ admits the following expression:
\begin{equation}\label{eq:loo}
    \ell^o(\theta) \cong \sum_{k=1}^N \frac{1}{2} \log| S_k(\theta)|+\frac{1}{2}r_k(\theta)^\top (S_k(\theta))^{-1}r_k(\theta),
\end{equation}
where $r_k$ is from~\eqref{eq:r} and $S_k$ is from~\eqref{eq:S}.\hfill$\blacksquare$
\end{lemma}
\begin{proof}
The state filter provides the Gaussian predictive distribution 
$p(x_k \mid \bm{y}_{1:k-1}^o,\theta)=\mathcal{N}(\bar{x}_k,\bar{P}_k^o)$.  
Substituting this into the factorization~\eqref{eq:decom} yields
\begin{align*}
    p(y_k|\bm{y}^o_{k-1},\theta)=\mathcal{N}( H_k \bar{x}_k, H_k\bar{P}_k^oH_k^\top+R_k(\theta)).
\end{align*}
Evaluating its negative log-likelihood and discarding terms independent of $\theta$ leads directly to Lemma~\ref{lem:llo}.
\end{proof}

\begin{lemma}[Computation of $\ell^s(\theta)$]\label{lem:lls}
The supervisory loss $\ell^s(\theta)$ takes the form
\begin{equation}\label{eq:lss}
    \ell^s(\theta)
  \cong 
    \frac{1}{2}\log |C(\theta)|
    + \frac{1}{2} v(\theta)^\top C(\theta)^{-1} v(\theta),
\end{equation}
where $v=y^s - H^s \hat{X}^s_N$,
$
C(\theta)
= H^s \hat{P}_N^s(\theta) (H^s)^\top + \Psi$.
    \hfill$\blacksquare$
\end{lemma}
\begin{proof}
At the final filtering step, the state filter provides the Gaussian posterior
$
p(X^s \mid \bm{y}^o,\theta)
    = \mathcal{N}(\hat{X}^s_N,\,\hat{P}^s_N).$
Substituting this distribution and the supervisory measurement model~\eqref{eq:super_mea} into~\eqref{eq:ls}, and taking its negative log-likelihood gives~\eqref{eq:lss}.
\end{proof}

In the following two sections, we derive the gradients of the loss functions using both forward and reverse differentiation, in order to perform gradient descent optimization.

\section{Forward Differentiation}
\label{sec:forward_diff}
The forward differentiation can be implemented along with the forward likelihood evaluation in section~\ref{sec:state filter}.
By differentiating the Kalman filter recursions~\eqref{eq:predx} to \eqref{eq:copyP} termwisely, it follows that
\begin{align}
\partial_j \bar{X}_{k} 
&= F_{k,0} ~\partial_j \hat{X}_{k-1},\quad \partial_j \bar{P}_k = F_{k,0} ~\partial_j \hat{P}_{k-1} F_{k,0}^\top + \partial_j Q_{k,0},\notag\\
\partial_j r_k 
&= - H_k~ \partial_j \bar{x}_k,\label{eq:drk}\\
\partial_j S_k 
&= H_k ~\partial_j \bar{P}_k^o H_k^\top + \partial_j R_k,\label{eq:dSk}\\
\partial_j K_k 
&= \partial_j \bar{P}_k H_{k,0}^\top S_k^{-1}
   - \bar{P}_k H_{k,0}^\top S_k^{-1}\,\partial_j S_k\, S_k^{-1},\notag\\
\partial_j \check{X}_k 
&= \partial_j \bar{X}_k 
   + (\partial_j K_k)\, r_k 
   + K_k\, \partial_j r_k,\label{eq:dXk}\\
\partial_j \check{P}_k 
&= (I - K_k H_{k,0})\,\partial_j \bar{P}_k
   - (\partial_j K_k) H_{k,0} \bar{P}_k\label{eq:dPk},\\
\partial_j \hat{X}_k &= J_k \partial_j \check{X}_k,\quad \partial_j \hat{P}_k = J_k \partial_j \check{P}_k J_k^\top.\notag
\end{align}

Invoking the chain rule to~\eqref{eq:loo} and~\eqref{eq:lss}, the next two lemmas establish the derivatives of $\ell^o(\theta)$ and $\ell^s(\theta)$.
\vspace{0.5em}

\begin{lemma}[Derivative of $\ell^o(\theta)$]\label{lem:dlo}
The derivative of $\ell^o(\theta)$ with respect to $\theta_j$ is given by
\begin{equation}
    \frac{\partial \ell^o(\theta)}{\partial \theta_j}
    = \sum_{k=1}^N 
      \frac{\partial l_k^o(\theta)}{\partial \theta_j},
\end{equation}
where 
$
l_k^o(\theta)
\triangleq 
\frac{1}{2}\log |S_k(\theta)|
+ \frac{1}{2} r_k(\theta)^\top S_k(\theta)^{-1} r_k(\theta).
$
Its derivative is
\begin{align*}
\frac{\partial l_k^o(\theta)}{\partial \theta_j}
&=
\frac{1}{2}\operatorname{tr}\!\left(S_k^{-1} \frac{\partial S_k}{\partial \theta_j}\right)
+ 
\left(\frac{\partial r_k}{\partial \theta_j}\right)^\top
S_k^{-1} r_k\\
&\quad- \frac{1}{2} 
r_k^\top S_k^{-1} 
\left(\frac{\partial S_k}{\partial \theta_j}\right)
S_k^{-1} r_k,
\end{align*}
where $\frac{\partial r_k}{\partial \theta_j}$ and $\frac{\partial S_k}{\partial \theta_j}$
are obtained from~\eqref{eq:drk} and~\eqref{eq:dSk}.
\hfill$\blacksquare$
\end{lemma}


\begin{lemma}[Derivative of $\ell^s(\theta)$]\label{lem:dls}
The derivative of $\ell^s(\theta)$ with respect to $\theta_j$ is given by
\begin{align}
\frac{\partial \ell^s(\theta)}{\partial \theta_j}
&=
\frac{1}{2}\operatorname{tr}\!\left(C^{-1} \frac{\partial C}{\partial \theta_j}\right)
+
\left(\frac{\partial v}{\partial \theta_j}\right)^\top C^{-1} v\\
&\quad- \frac{1}{2}
v^\top C^{-1}
\left(\frac{\partial C}{\partial \theta_j}\right)
C^{-1} v,
\end{align}
where 
$
\frac{\partial C}{\partial \theta_j} 
= H^s \left(\frac{\partial \hat{P}_N^s}{\partial \theta_j}\right) (H^s)^\top,
\quad
\frac{\partial v}{\partial \theta_j}
= - H^s \frac{\partial \hat{X}_N^s}{\partial \theta_j},
$
The quantities $\frac{\partial \hat{X}^s_N}{\partial \theta_j}$ and $\frac{\partial \hat{P}_N^s}{\partial \theta_j}$ are obtained from the
sensitivity recursions in~\eqref{eq:dXk} and~\eqref{eq:dPk}.
\hfill$\blacksquare$
\end{lemma}

\begin{theorem}[Forward Differentiation of $\mathcal{L}(\theta)$]
\label{the:L}
The derivative of $\mathcal{L}(\theta)$ is
\begin{equation}
    \frac{\partial \mathcal{L}(\theta)}{\partial ~\theta_j}  = \frac{\partial \ell^o(\theta)}{\partial \theta_j} + \frac{\partial \ell^s(\theta)}{\partial \theta_j},
\end{equation}
where $\frac{\partial \ell^o(\theta)}{\partial \theta_j}$ and $\frac{\partial \ell^s(\theta)}{\partial \theta_j}$ 
are from Lemma~\ref{lem:dlo} and~\ref{lem:dls}.
\hfill$\blacksquare$
\end{theorem}

\begin{figure*}[!b] 
\noindent\rule{\textwidth}{0.4pt}\par
\vspace{-0.5em}
\setlength{\jot}{0.5pt} 
\noindent\small

\begin{align*}
\frac{\partial \mathcal{L}}{\partial \check{X}_k}&=
J_k^\top \frac{\partial \mathcal{L}}{\partial \hat{X}_k},\quad
\frac{\partial \mathcal{L}}{\partial \check{P}_k}= J_k^\top \frac{\partial \mathcal{L}}{\partial \hat{P}_k} J_k,
\notag\\
\frac{\partial \mathcal{L}}{\partial \bar{X}_k} &= (I-K_kH_{k,0})^\top \frac{\partial \mathcal{L}}{\partial \check{X}_{k}} -H_{k,0}^\top S_k^{-1} r_k  ,\\
\frac{\partial \mathcal{L}}{\partial \bar{P}_k }& = \left(I-K_k H_{k,0}\right)^\top \left[ \frac{\partial \mathcal{L}}{\partial \check P_k}+\frac{1}{2}\frac{\partial \mathcal{L}}{\partial \check{X}_k}r_k^\top R_k^{-1}H_{k,0} + \frac{1}{2}H_{k,0}^\top R_k^{-1}r_k \left(\frac{\partial \mathcal{L}}{\partial \check{X}_k}\right)^\top\right]\left(I-K_k H_{k,0}\right)+\frac{1}{2}\left(H_{k,0}^\top S_k^{-1} H_{k,0} -H_{k,0}^{\top} S_k^{-1} r_k r_k^\top S_k^{-1} H_{k,0} \right),\notag\\
\frac{\partial \mathcal{L}}{\partial \hat{X}_{k-1}} &=F_{k,0}^\top  \frac{\partial \mathcal{L}}{\partial \bar{X}_k},\quad\frac{\partial \mathcal{L}}{\partial \hat{P}_{k-1}} = F_{k,0}^\top \frac{\partial \mathcal{L}}{\partial \bar{P}_{k}} F_{k,0} ,\\
\frac{\partial \mathcal{L}}{\partial R_k} &=K_k^\top \frac{\partial \mathcal{L}}{\partial \check{P}_k} K_k-\frac{1}{2} K_k^\top \frac{\partial \mathcal{L}}{\partial \check{X}_k}r_k^T S_k^{-1}-\frac{1}{2} S_k^{-1} r_k\left(\frac{\partial \mathcal{L}}{\partial \check{X}_k}\right)^\top K_k +\frac{1}{2}S_k^{-1}-\frac{1}{2}S_k^{-1} r_k r_k^\top S_k^{-1}\label{eq:L_Rk}, \quad 
\frac{\partial \mathcal{L}}{\partial Q_k}=\frac{\partial \mathcal{L}}{\partial \bar{P}_k }.
\end{align*}

\end{figure*}

\section{Reverse Differentiation}
\label{sec:reverse_diff}
Reverse-mode differentiation is performed after the forward likelihood evaluation, propagating sensitivities backward using the quantities stored during the forward pass. The dependency diagram provided in Fig.~\ref{fig:CG} helps organize this backward computation.
\begin{figure}[ht]
    \centering
    \includegraphics[width=0.6\linewidth]{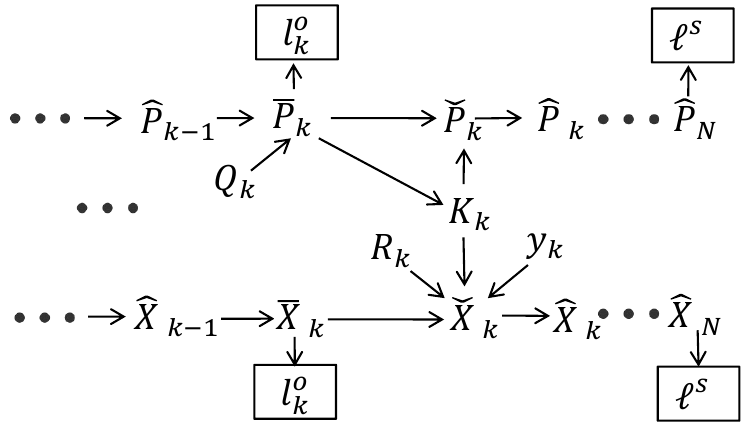}
    \caption{Dependency diagram.}
    \label{fig:CG}
\end{figure}

The following theorem is a direct consequence of Theorem~1 in~\cite{parellier2023speeding}. 
For details, we refer the reader to the original proof.

\vspace{0.5em
}
\begin{theorem}[Reverse Differentiation of $\mathcal{L}(\theta)$.]
\label{theo:re}
The derivative of $\mathcal{L}(\theta)$ with respect to $\theta$ is computed via a backward propagation procedure.
This process is initialized at $k=N$ by setting
\begin{align*}
\frac{\partial \mathcal{L}}{\partial \hat{X}_N^s}&=\frac{\partial \ell^s}{\partial \hat{X}_N^s} = -(H^s)^\top C^{-1} v,\\
\frac{\partial \mathcal{L}}{\partial \hat{P}_N^s}&=\frac{\partial \ell^s}{\partial \hat{P}_N^s} = \frac{1}{2}\left((H^s)^\top C^{-1}H^s -(H^s)^\top C^{-1}v v^\top C^{-1} H^s\right),\\
\frac{\partial \mathcal{L}}{\partial \hat{X}_N}&=G^\top \frac{\partial \mathcal{L}}{\partial \hat{X}_N^s},\quad \frac{\partial \mathcal{L}}{\partial \hat{P}_N}=G^\top \frac{\partial \mathcal{L}}{\partial \hat{P}_N^s} G,
\end{align*}
where $G = [\bm0\;I]$ to extract $\hat{X}_N^s$ from $\hat{X}_N$.
The gradients are propagated recursively backward to the desired time step $k \leq N$ using the relations given at the bottom of the page.
\hfill$\blacksquare$
\end{theorem}

The analytical gradient of $\mathcal{L}(\theta)$, obtained via forward or reverse differentiation, supports gradient-based updates of $\theta$. Algorithms~\ref{alg:forward_mode} and~\ref{alg:reverse_mode} summarize the whole procedures.
\begin{algorithm}[t]
\caption{Forward-Mode Covariance Estimation}
\label{alg:forward_mode}
\begin{algorithmic}[1]
\State \textbf{Input:} Initial parameter $\theta_0$; initial state $(x_0, P_0)$; controls $\{u_k\}_{k=1}^N$; measurements $(\bm y^o, \bm y^s)$; \texttt{itermax}.
\State $\hat{\theta} \gets \theta_0$
\For{$i = 1:\texttt{itermax}$}
    \For{$k = 1:N$}
        \State Perform the state filter recursion~\eqref{eq:predx}--\eqref{eq:copyP}.
        \State Compute $\partial_{\theta_j} l_k^o(\hat{\theta})$ using Lemma~\ref{lem:dlo}.
    \EndFor
    \State Compute $\partial_{\theta_j} \ell^s(\hat{\theta})$ using Lemma~\ref{lem:dls}.
    \State Compute $\nabla_{\theta} \mathcal{L}(\hat{\theta})$ using Theorem~\ref{the:L}.
    \State $\hat{\theta} \gets \hat{\theta} - \eta_i \, \nabla_{\theta} \mathcal{L}(\hat{\theta})$.
\EndFor
\end{algorithmic}
\end{algorithm}
\vspace{-1em}
\begin{algorithm}[t]
\caption{Reverse-Mode Covariance Estimation}
\label{alg:reverse_mode}
\begin{algorithmic}[1]
\State \textbf{Input:} Initial parameter $\theta_0$; initial state $(x_0, P_0)$; controls $\{u_k\}_{k=1}^N$; measurements $(\bm y^o, \bm y^s)$; \texttt{itermax}.
\State $\hat{\theta} \gets \theta_0$
\For{$i = 1:\texttt{itermax}$}
    \For{$k = 1:N$}
        \State Perform the state filter recursion~\eqref{eq:predx}--\eqref{eq:copyP}.
    \EndFor
    \For{$k = N:1$}
        \State Implement recursion in ~\ref{theo:re} to compute 
$\nabla_{\theta} \mathcal{L}(\hat{\theta})$.
    \EndFor
    \State $\hat{\theta} \gets \hat{\theta} - \eta_i \, \nabla_{\theta} \mathcal{L}(\hat{\theta})$.
\EndFor
\end{algorithmic}
\end{algorithm}

\section{Discussion}
In this section, we compare our approach with the closely related work \cite{parellier2023speeding}. The work~\cite{parellier2023speeding} derives closed-form reverse-mode derivatives for a scalar loss of the form $L(\theta)=\sum_{k=1}^N (l_{k|k-1}+l_{k|k})$, of which the two terms depend on the prior and posterior filter variables. Our total log likelihood follows the same structure, with $\ell^o$ and $\ell^s$ playing analogous roles. Unlike the general formulation in \cite{parellier2023speeding}, our decomposition is derived directly from a probabilistic factorization of the likelihood, which yields a principled balance between the primary and supervisory losses. Moreover, while \cite{parellier2023speeding} highlights the computational efficiency of reverse-mode differentiation, it does not address the associated memory costs. We analyze these trade-offs in the following.

The key distinction between the two modes lies in how the chain rule is applied.
Let $D$ be the augmented state dimension, that is, $\hat{X}_N \in \mathbb{R}^D$. 
The forward mode propagates a $p$-dimensional sensitivity vector through the $N$ Kalman recursions, yielding a cost of $\mathcal{O}(pND^{3})$.
Reverse mode instead backpropagates scalar adjoints from the loss, producing all partial derivatives in one sweep with cost $\mathcal{O}(ND^{3})$, comparable to that of a single Kalman filter pass.
However, reverse-mode differentiation incurs high memory overhead because all intermediate quantities must be stored until the backward pass is complete, and gradients are computed only after the forward likelihood evaluation finishes.
The forward mode, on the other hand, requires no storage of past states, computes derivatives concurrently with the Kalman recursion, and naturally supports parallelization across parameter dimensions.

In summary, differentiating through a Kalman filter must account for its inherently sequential structure. Forward differentiation offers low memory usage and is well-suited to online or resource-constrained settings, but its cost scales with the parameter dimension. Reverse differentiation eliminates this scaling at the expense of storing all intermediate states. Neither method is universally superior—the choice depends on the computational context. These trade-offs are summarized in Table~\ref{tab:derivative_methods}.

\begin{table*}[t]
    \centering
    \begin{tabular}{c|c|c|c|c}
    \hline
        Derivative method & Saved variables & Space complexity  & Time complexity & Suggestion \\
    \hline
        Forward differentiation 
        & current-step variables 
        & $\mathcal{O}(pD^2)$ 
        & $\mathcal{O}(pND^3)$ 
        & online adaptation or low-dimensional $\theta$ \\
    \hline
        Reverse differentiation 
        & $\{F_k,K_k, H_k, S_k, r_k\}_{k=1}^N$ 
        & $\mathcal{O}(ND^2)$ 
        & $\mathcal{O}(ND^3)$ 
        & offline optimization or high-dimensional $\theta$ \\
    \hline
    \end{tabular} 
    \caption{Comparison: forward- and reverse-differentiation for Kalman-filter contributing to $\mathcal{L}(\theta)$.}
    \label{tab:derivative_methods}
\end{table*}

\section{Simulation and Analysis}
\subsection{Simulation Settings}
The same linear system model as~\cite{parellier2023speeding} is
used. 
The state consists of a six-dimensional vector $\left[p_k^\top~v_k^\top\right]^\top$, where the position $p_k \in \mathbb{R}^3$ and the velocity $v_k\in \mathbb{R}^3$.
The noisy accelerations are considered as inputs $u_k \in \mathbb{R}^3$.
The dynamics write
\begin{align*}
    p_k &= p_{k-1} + \Delta t~ v_{k-1} + \omega_{k}^p,\\
    v_k & =v_{k-1} + \Delta t~ u_{k-1} + \omega_{k}^v,
\end{align*}
where $\Delta t=1$ is the time step and $\omega_k\sim \mathcal{N}(\bm0,Q)$ is the process noise with $Q=qI$ and $q=0.01$.

The measurement $y_k$ is the noisy 3-D position
\begin{equation*}
    y_k = p_k + \nu_k,\qquad 
    \nu_k \sim \mathcal{N}(\bm 0, R(\theta)).
\end{equation*}

In this example, the supervisory measurements consist of relative positions between selected pairs of supervisory states. 
For any pair of state $x_i$ and $x_j$,
\begin{equation*}
    y^s_{ij} = p_i-p_j+\nu_{ij}^s,  \quad \nu_{ij}\sim \mathcal{N}(\bm 0, \alpha I)
\end{equation*}
Stacking all the $y^s_{ij}$ yields $\bm{y}^s$ and the overall supervisory noise is $\Psi=\alpha I$, with $\alpha=0.01$. 

The goal is to optimize $R(\theta)$ by minimizing the proposed $\mathcal{L}(\theta)$ on the calibration trajectory, and to assess the learned parameter on independent test trajectories. Two trajectories are discussed as follows.

\textbf{Calibration trajectory.}
The calibration trajectory contains 100 steps, as illustrated in Fig.~\ref{fig:calib_traj}. 
Supervisory measurements are generated using a downsampling rate and a distance threshold.
\begin{figure}[ht]
    \centering
    \includegraphics[width=0.6\linewidth]{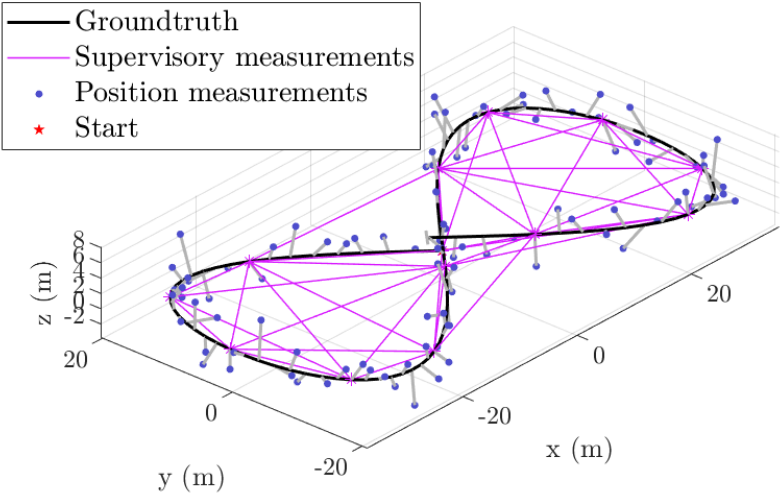}
    \caption{Calibration trajectory.}
    \label{fig:calib_traj}
\end{figure}

\textbf{Test trajectory.}
The test dataset consists of 600 steps, and the corresponding path is shown in Fig.~\ref{fig:test_traj}.
\begin{figure}[h]
    \centering
    \includegraphics[width=0.6\linewidth]{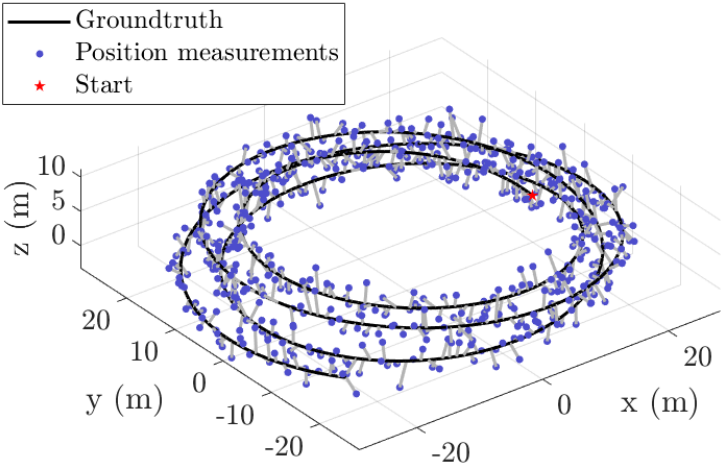}
    \caption{Test trajectory.}
    \label{fig:test_traj}
\end{figure}

\textbf{Simulation noise setup.}
Following~\cite{khosoussi2025joint}, the measurement noise $\nu_k\sim \mathcal{N}(\bm 0, R^{\text{true}})$ is given by
$
    R^{\text{true}} =  R_{\text{base}} + R_d,
$
where $ R_{\text{base}}$ is a fixed covariance matrix that encodes the correlation of measurement noises, and $R_d=\operatorname{diag}(0.9^2, 1.3^2, 2.2^2)~m^2$, which specifies the noise level along each axis. The resulting GPS noise has an RMSE of approximately 2.9~m.

\subsection{Results}
We compare our method with a $\ell^o$-only case, which refers to methods in~\cite{parellier2023speeding}. 
We also compare three parameterizations
scheme of $R(\theta)$, as listed below:
To ensure the positive defiteness of $R(\theta)$, we test three paramterization methods: 
\begin{enumerate}
    \item Isotropic covariance ($\theta \in \mathbb{R}$):
    $
    R(\theta) = \exp(\theta) I.
    $
    \item Diagonal matrix ($\theta \in \mathbb{R}^3$):
    $$
    R(\theta) = \operatorname{diag}(\exp(\theta_1)\;  \exp(\theta_2)\;\exp(\theta_3)\;).
    $$
    \item Cholesky decomposition ($\theta \in \mathbb{R}^6$):
    $$
    R(\theta)=L(\theta)L(\theta)^\top,
    $$
    where $L(\theta)$ is a low-triangular matrix, with positive diagonal entries enforced through exponentiation. 
\end{enumerate}
Methods 1 and 2 use forward differentiation due to their small parameter dimension, whereas Method 3 adopts reverse mode. 
All methods share the same optimization setup and run for 20 iterations, with $\theta$ optimized on the calibration trajectory. 
Fig. \ref{fig:LL} shows the loss evolution across iterations, along with the corresponding RMSE curves that reflect the tuning performance.
\begin{figure}[ht]
    \centering
    \includegraphics[width=0.7\linewidth]{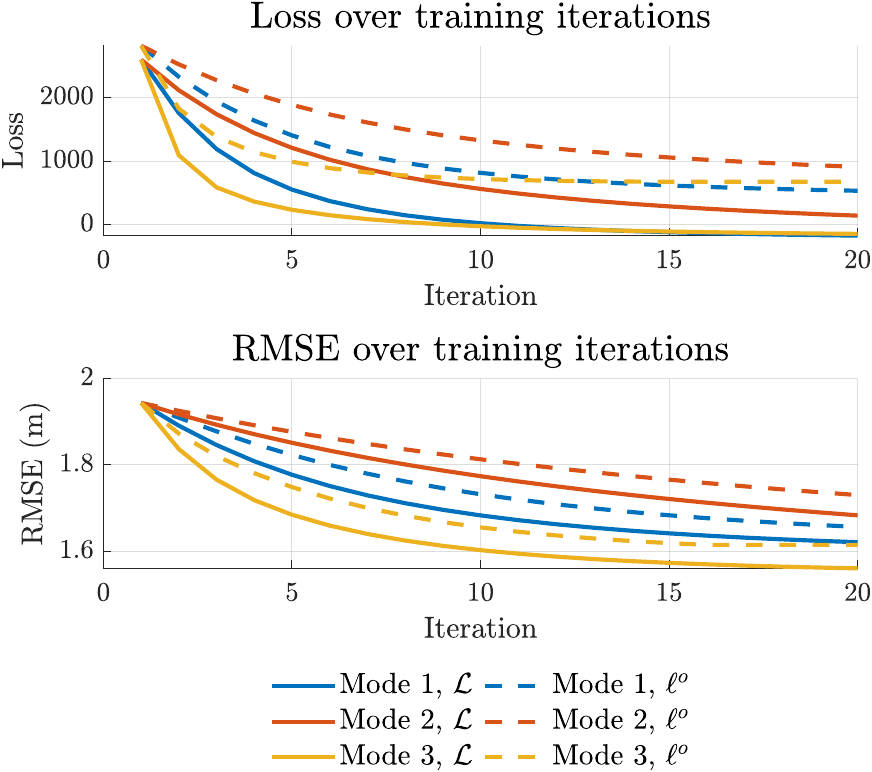}
    \caption{Loss and RMSE over the training iterations. }
    \label{fig:LL}
\end{figure}

The estimated parameters are evaluated on an independent testing trajectory. 
We conduct 100 Monte Carlo trials using the same noise model but independently sampled calibration and testing data. 
The average testing results are summarized in the table, where parentheses indicate the number of supervisory states and supervisory measurements. 
Labels (F) and (R) denote gradients computed via forward- and reverse-mode differentiation.
\begin{table}[ht]
    \centering
    \begin{tabular}{c|ccc}
    \hline
    \textbf{Method} & $\ell^o$ & $\mathcal{L^-} (13,40)$ & $\mathcal{L^+} (20, 147)$ \\
    \hline
    Mode 1~(F) &  1.6091  &1.6038    & 1.5795  \\
    Mode 2~(F) & 1.6669   & 1.6456   & 1.6287  \\
    Mode 3~(R) & 1.6148   & 1.5838   & \textbf{1.5606}  \\
    \hline
    \end{tabular}
   \caption{Average RMSE~(m) on the test trajectory over 100 Monte-Carlo trials. }
    \label{tab:rmse}
\end{table}

In summary, the simulation used relatively large measurement noise, and MLE-based noise tuning consistently reduced state estimation error. 
Incorporating supervisory loss further improved optimization, and additional supervisory measurements generally helped, though at increased computational cost due to state augmentation.
These results suggest several practical guidelines.
First, the parameterization should match the intrinsic noise structure, as the Cholesky form better captures the structure of the covariance, yielding the best performance, and more degrees of freedom do not necessarily yield better performance. For example, Mode 1 outperforms Mode 2 despite its lower dimensionality.
Second, the choice of differentiation method should reflect the parameter dimensionality: forward mode is more memory-efficient for small parameter spaces, whereas reverse mode becomes advantageous as dimensionality increases. 
Finally, supervisory information helps, but its size should be limited. For instance, $n_s$ supervisory states can form up to $n_s(n_s-1)/2$ mutual observations to provide strong supervision while keeping the state manageable in the state filter.

\section{Conclusion}
This paper studied an MLE/MAP framework for estimating process and measurement noise covariances in linear Gaussian systems. By factorizing the likelihood of primary and supervisory measurements, the approach enables efficient state estimation and analytic gradient computation through both forward- and reverse-mode differentiation. The maximum-likelihood optimization reliably improves estimation accuracy, with performance depending critically on appropriate noise parameterization, differentiation strategy, and the controlled use of supervisory information. Beyond the linear setting, the framework is readily extensible—for example, to nonlinear state-space models or to hybrid formulations that incorporate deep learning components for noise modeling or adaptive parameterization. Overall, the proposed method provides a principled and computationally sound foundation for noise covariance estimation with clear potential for broader applications.


\bibliography{ifacconf}             
        
\end{document}